\documentclass[11pt]{amsproc}
\usepackage{mathrsfs}
\usepackage{stmaryrd}
\usepackage{cases}
\usepackage{amssymb}
\usepackage{amsmath}
\usepackage{amsfonts}
\usepackage{graphicx}
\usepackage{amsmath,amstext,amsbsy,amssymb}
\usepackage{verbatim}
\newtheorem{theorem}{Theorem}[section]

\theoremstyle{definition}

\theoremstyle{remark}

\numberwithin{equation}{section} \errorcontextlines=0

\newcommand{\la}{\lambda}

\begin{document}
\title[Quantum Discords of Tripartite Quantum Systems ]
{Quantum Discords of Tripartite Quantum Systems  }
\author{Jianming Zhou}
\address{Zhou: Department of Mathematics, Shanghai University, Shanghai 200444, China}
\email{272410225@qq.com}

\author{Xiaoli Hu*}
\address{Hu: School of Artificial Intelligence, Jianghan University, Wuhan, Hubei 430056, China}
\email{xiaolihumath@jhun.edu.cn (Corresponding author)}

\author{Naihuan Jing}
\address{Jing: Department of Mathematics,
   North Carolina State University,
   Raleigh, NC 27695, USA}
\email{jing@ncsu.edu}
%\footnotetext[1]{Corresponding author.}
\thanks{*Corresponding author: Xiaoli Hu (xiaolihumath@jhun.edu.cn)}


\begin{abstract}
 The quantum discord of bipartite systems is one of the best-known measures of non-classical correlations and an important quantum resource. In the recent work appeared in [Phys. Rev. Lett 2020, 124:110401], the quantum discord has been generalized to multipartite systems. In this paper, we give analytic solutions of the quantum discord for tripartite states with fourteen parameters.

%PACS numbers: 03.65.Ud, 03.67.Mn, 02.10.Yn, 02.10.Xm
\end{abstract}
\maketitle
\section{Introduction}

The quantum discord usually involves with quantum entanglement and umentangled quantum correlations in quantum systems. It measures the total non-classical correlation in a quantum system, and has attracted widespread attention since its appearance. Applications of the non-entanglement quantum correlations in quantum information processings have been extensively studied, including the quantum computing scheme of DQC1 \cite{DSC} and Grover search algorithm \cite{CF} etc. This partly explains why quantum schemes surpass classical schemes. Meanwhile, the quantum discord as a non-classical correlation is one of the important quantum resources and is ubiquitous in many areas of modern physics ranging from condensed matter physics, quantum optics, high-energy physics to quantum chemistry, thus can be regarded as one of the fundamental non-classical correlations besides entanglement and EPR-steerable states \cite{OZ,L}.

The quantum discord is defined as the maximal difference between the quantum mutual information without and with a von Neumann projective measurement applying to one part of the bipartite system. For tripartite and lager systems, some generalizations of the discord have been proposed \cite{RS,OW,GBG,CAP,MPS, RLB}, and have been used in quantum information processings. It is well-known that quantum discord is extremely difficult to evaluate and most exact solutions are only for the X-type quantum states (cf. \cite{ZJW, St, YF, JY}). This paper is devoted to quantification of the quantum correlation in tripartite and larger systems to derive some exact solutions for non-X-type states, and we hope it can contribute to better understanding and more effective use of quantum states in realizing quantum information processing schemes.

The paper is organized as follows. We first introduce the generalized discord for tripartite systems \cite{RLB} based on that of bipartite systems \cite{OZ}. We derive analytic solutions for tripartite states with fourteen parameters. Furthermore, the quantum discord of some well-known states (such as GHZ states) are computed.
\section{Generalization of quantum discord to tripartite states}

For a bipartite state $\phi^{bc}$ on system $H_B\otimes H_C$, the quantum mutual information is $I(\phi^{bc}):=S_B(\phi^b)+S_C(\phi^c)-S_{BC}(\phi^{bc})$, where $S(\phi^{X})=\mathrm{Tr}\phi^{X} \log_2(\phi^{X})$ is the von Neumann entropy of the quantum state on system X. Set $\{\Pi^B_k\}$ to be an one-dimensional von Neumann projection operator on subsystem $B$ which satisfies $\sum_k\Pi_k^B=I, (\Pi_k^B)^2=\Pi_k^B,\Pi_k^B\Pi_{k'}^B=\delta_{kk^{'}}$. Then the state $\phi^{bc}$ under the measurement $\{\Pi^B_k\}$ is changed into
$$\phi^c_k=\frac{1}{p_k}\mathrm{Tr}_B(I\otimes \Pi^B_k)\phi^{bc}(I\otimes \Pi^B_k)$$
with the probability $p_k=\mathrm{Tr}(I\otimes \Pi^B_k)\phi^{bc}(I\otimes \Pi^B_k)$. For simplicity, we denote by $\Pi^X$ the measurement $\{\Pi_k^X\}$ on system $X$. The quantum conditional entropy is simply given by $S_{C|\Pi^B}(\phi^{bc})=\sum_kp_kS(\phi^c_k)$.
Then the measurement-induced quantum mutual information is given by
$$C(\phi^{bc})=S_C(\phi^c)-\min S_{C|\Pi^B}(\phi^{bc}).$$
By Olliver and Zurek \cite{OZ}, the original definition of the quantum discord $Q(\rho)$ is the difference of the quantum mutual information $I(\phi^{bc})$ and the measurement-induced quantum mutual information $C(\phi^{bc})$, i.e.
\begin{equation}
Q(\phi^{bc})=I(\phi^{bc})-C(\phi^{bc})=\min_{\Pi^B}\{S_{C|\Pi^B}(\phi^{bc})-S_{C|B}(\phi^{bc})\},\label{eq1}
\end{equation}
where $S_{C|B}(\phi^{bc})=S_{BC}(\phi^{bc})-S_B(\phi^b)$ is the unmeasured conditional state on subsystem $C$.

For the tripartite system $H_A\otimes H_B\otimes H_C$, we consider the $BC$ composite system as the first subsystem and $A$-system as the second subsystem. The state $\rho^{abc}$ of system $H_{A}\otimes H_{B}\otimes H_C$ gives arise to a state on $BC$-subsystem after the von Neumann measurement $\{\Pi_j^A\}$ on $A$ subsystem. Namely, it takes the following form:
\begin{equation}
\begin{split}
\rho^{bc}_j=\frac{1}{p^{bc}_j}\mathrm{Tr}_A(\Pi_j^A\otimes I)\rho^{abc}(\Pi_j^A\otimes I)
\end{split}
\end{equation}
with probability $p^{bc}_j=\mathrm{Tr}(\Pi_j^A\otimes I)\rho^{abc}(\Pi_j^A\otimes I)$.The measured quantum mutual information of $\rho^{abc}$ is naturally given by
\begin{equation}
\begin{split}
\mathcal{J}(\rho^{abc}|\Pi^A)=S_{BC}(\rho^{bc})-S_{BC|\Pi^A}(\rho^{abc}).
\end{split}
\end{equation}
The quantity of classical correlation of the tripartite state $\rho^{abc}$ is
\begin{equation}
\begin{split}
\mathcal{C}(\rho^{abc})=\max_{\Pi^A}\mathcal{J}(\rho^{abc}|\Pi^A)=S_{BC}(\rho^{bc})-\min_{\Pi^A} S_{BC|\Pi^A}(\rho^{abc}).
\end{split}
\end{equation}
We know that the quantum mutual information $I(\rho^{abc})=S_A(\rho^a)+S_{BC}(\rho^{bc})-S_{ABC}(\rho^{abc})$. Similar to Eq.(\ref{eq1}), the generalized quantum discord of the tripartite state $\rho^{abc}$ can be defined as
\begin{equation}
\begin{split}
\mathcal{Q}(\rho^{abc})=I(\rho^{abc})-C(\rho^{abc})=\min_{\Pi^A}\{S_{BC|\Pi^A}(\rho^{abc})-S_{BC|A}(\rho^{abc})\},
\end{split}
\end{equation}
where $S_{BC|A}(\rho^{abc})=S_{ABC}(\rho^{abc})-S_{A}(\rho^a)$ is the unmeasured conditional entropy on $BC$-bipartite subsystem.

In order to evaluate the quantity $\min_{\Pi^A}S_{BC|\Pi^A}(\rho^{abc})$, the multipartite measurement based on conditional operators can be constructed as follows: \cite{MM}
\begin{equation}
\begin{split}
\Pi_{jk}^{AB}=\Pi_j^A\otimes\Pi_{k|j}^B
\end{split}
\end{equation}
with the measurement ordering from $A$ to $B$. The projector $\Pi_{k|j}^B$ on subsystem $B$ is conditional measurement outcome of $A$. These projectors satisfy $\sum_k\Pi^B_{k|j}=I^B, \sum_j\Pi_j^A=I^A$.
%Then, the projection operator $\Pi^B_{k|j}$ is used to measure the subsystem B of the system BC measured by the projection operator $\Pi_j^A$.
Then after the measurement $\Pi_{jk}^{AB}$, the state $\rho^{abc}$ is collapsed to a state on subsystem $C$, i.e.
\begin{equation}
\begin{split}
\rho^c_{jk}=\frac{1}{p^c_{jk}}\mathrm{Tr}_{AB}(\Pi_{jk}^{AB}\otimes I)\rho^{abc}(\Pi_{jk}^{AB}\otimes I)
\end{split}
\end{equation}
with the probability $p^c_{jk}=\mathrm{Tr}(\Pi_{jk}^{AB}\otimes I)\rho^{abc}(\Pi_{jk}^{AB}\otimes I)$. The conditional entropy after the $AB$-bipartite measurement is
$$S_{C|\Pi^{AB}}(\rho^{abc})=\sum_{jkl}p_{jk}^{c}\lambda_l^{(jk)}\log_2\lambda_l^{(jk)},$$
where $\lambda_l^{(jk)}$ are eigenvalues of state $\rho^c_{jk}$.

%For the bipartite state $\rho_j^{bc}$ on $BC$-subsystem, the state under the partial trace on subsystem $C$ is changed into
%\begin{equation}
%\begin{split}
%\rho^{b}_j=\mathrm{Tr}_C[\frac{1}{p^b_j}\mathrm{Tr}_A(\Pi_j^A\otimes I)\rho^{abc}(\Pi_j^A\otimes I)].
%\end{split}
%\end{equation}
%with the probability $\mathrm{Tr}\{\mathrm{Tr}_C[\frac{1}{p^b_j}\mathrm{Tr}_A(\Pi_j^A\otimes I)\rho^{abc}(\Pi_j^A\otimes I)]\}$. We see that $p^{b}_j=p^{bc}_j$.
Let $\rho_{\Pi^X}=\sum_{\Pi^X}\Pi^X\rho\Pi^X$ be the state after measurement $\Pi^X$. Then for a bipartite state $\rho^{ab}$, the conditional entropy on subsystem $B$ after the measurement on subsystem $A$ is
\begin{equation}
\begin{split}
S_{B|\Pi^A}(\rho^{ab})=\sum_jp_jS_{B}(\rho^b_j).
\end{split}
\end{equation}
%Then we have $S_{BC|\Pi^A}(\rho)=S_{C|\Pi^{AB}}(\rho)+S_{B|\Pi^A}(\rho)$.
By  \cite[Eq.(6)]{RLB}, the entropy of the measured system can always be decomposed as
\begin{equation}S_{AB}(\rho^{ab}_{\Pi^A})=S_A(\rho^{ab}_{\Pi^A})+S_{B|\Pi^A}(\rho^{ab}). \label{S(AB)}
\end{equation}
For the tripartite system, using the measurement $\Pi^{AB}$, we have
\begin{equation} S_{ABC}(\rho^{abc}_{\Pi^{AB}})-S_{AB}(\rho^{abc}_{\Pi^{AB}})=S_{C|\Pi^{AB}}(\rho^{abc}),\label{S(ABC1)}
\end{equation}
when the measurement on $A$ system is $\Pi^A$, then we have
\begin{equation} S_{ABC}(\rho^{abc}_{\Pi^{A}})-S_{A}(\rho^{abc}_{\Pi^{A}})=S_{BC|\Pi^{A}}(\rho^{abc}).\label{S(ABC2)}
\end{equation}
By Eq.(\ref{S(AB)}), Eq.(\ref{S(ABC1)}), Eq.(\ref{S(ABC2)}), we have that
 $$S_{BC|\Pi^A}(\rho^{abc})=S_{B|\Pi^A}(\rho^{ab})+S_{C|\Pi^{AB}}(\rho^{abc}).$$
Meanwhile, $S_A(\rho^{abc}_{\Pi^{AB}})=S_A(\rho^{abc}_{\Pi^{A}})$, so the generalization discord of a tripartite state can be written as \cite{RLB}
\begin{equation}\label{1.2}
\begin{split}
\mathcal{Q}(\rho):=\min_{\Pi^{AB}}[-S_{BC|A}(\rho)+S_{B|\Pi^A}(\rho)+S_{C|\Pi^{AB}}(\rho)].
\end{split}
\end{equation}

\section{Quantum Discord of non-X Qubit-Qutrit state}
For the product states in the tripartite system, the discord has the special property that it reduces to the standard bipartite discord when only bipartite quantum correlations are present. This means $\mathcal{Q}_{ABC}(\rho^{x}\otimes\rho^y)=\mathcal{Q}_{X}(\rho^{x})$ for $X=AB, BC$ and $AC$ subsystem. We consider the following tripartite states
\begin{equation}\label{eq:3.1}
\begin{split}
\rho^{abc}=&\frac{1}{8}(I_8+a_3\sigma_3\otimes I_4+I_2\otimes b_3\sigma_3\otimes I_2+I_4\otimes\sum_i^3c_i\sigma_i
\\+&\sum_i^3r_i\sigma_i\otimes\sigma_i\otimes I_2+\sum_i^3s_i\sigma_i\otimes I_2\otimes\sigma_i
+\sum_i^3T_i\sigma_i\otimes\sigma_i\otimes\sigma_i),
\end{split}
\end{equation}
where $I_d$ represents the unit matrix of order $d$, and $\sigma_i(i=1,2,3)$ are Pauli matrices. The parameters $a_3, b_3, c_i, r_i, s_i, T_i\in \mathbb{R}$ and they are confined within the internal $[-1,1]$. Its matrix has the following form:
\begin{equation}
\begin{split}
\rho=\left(
  \begin{array}{cccccccc}
    * & * & 0 & 0 & 0 & * & * & * \\
    * & * & 0 & 0 & * & 0 & * & * \\
    0 & 0 & * & * & * & * & 0 & * \\
    0 & 0 & * & * & * & * & * & 0 \\
    0 & * & * & * & * & * & 0 & 0 \\
    * & 0 & * & * & * & * & 0 & 0 \\
    * & * & 0 & * & 0 & 0 & * & * \\
    * & * & * & 0 & 0 & 0 & * & * \\
  \end{array}
\right).
\end{split}
\end{equation}

Let $\{|j\rangle\langle j|, j=0,1\}$ be the computational base, then any von Neumann measurement on system $X$ can be written as $\{\Pi_j^X=V|j\rangle\langle j|V^{\dagger}|, j=0,1\}$ for some unitary matrix $V\in \mathrm{SU}(2)$. Any unitary matrix can be written as $V=tI+\sqrt{-1}\sum_ky_k\sigma_k$ with $t, y_k\in \mathbb{R}$.
 When the measurement $\{\Pi^X_j\}$ is performed locally on one part of the composite system $Y\otimes X$, the ensemble $\{\rho_j^Y,p_j^Y\}$ is given by
 $\rho_j^Y=\frac{1}{p_j^Y}\mathrm{Tr}_X(I\otimes\Pi^X_j)\rho^{YX}(I\otimes \Pi^X_j)$ with the probability $p_j^Y=\mathrm{Tr}[\rho^{YX}(I\otimes \Pi_j^X)]$.

It follows from symmetry that
\small\begin{equation}
\begin{split}
V^{\dagger}\sigma_1V&=(t^2+y_1^2-y_2^2-y_3^2)\sigma_1+2(ty_3+y_1y_2)\sigma_2+2(-ty_2+y_1y_3)\sigma_3,\\
V^{\dagger}\sigma_2V&=(t^2+y_2^2-y_3^2-y_1^2)\sigma_1+2(ty_1+y_2y_3)\sigma_3+2(-ty_3+y_1y_2)\sigma_1,\\
V^{\dagger}\sigma_3V&=(t^2+y_3^2-y_1^2-y_2^2)\sigma_1+2(ty_2+y_1y_3)\sigma_3+2(-ty_1+y_2y_3)\sigma_2.\\
\end{split}
\end{equation}
Introduce new variables $z_1^X=2(-ty_2+y_1y_3),z_2^X=2(ty_1+y_2y_3),z_3^X=(t^2+y_3^2-y_1^2-y_2^2)$, then $(z_1^X)^2+(z_2^X)^2+(z_3^X)^2=1$.
Therefore $\Pi_j^X\sigma_k\Pi_j^X=(-1)^jz_k^X\Pi_j^X$ for $j=0,1$ and $k=1,2,3$.

For the tripartite state $\rho^{abc}$, the conditional state on $BC$ subsystem after measurement $\{\Pi_j^A  (j=0,1)\}$ on subsystem $A$ is
\begin{equation}
\begin{split}
\rho_j^{bc}&=\frac{1}{p_j^{bc}}(\Pi_1^A\otimes I_2\otimes I_2)\rho^{abc}(\Pi_1^A\otimes I_2\otimes I_2)\\
&=\frac{1}{p_j^{bc}}[(1+(-1)^ja_3z_3^A)I_2\otimes I_2+b_3\sigma_3\otimes I_2+(-1)^j\sum_i^3r_iz_i^A\sigma_i\otimes I_2\\
&+\sum_i^3(c_i+(-1)^js_iz_i^A)I_2\otimes\sigma_i+(-1)^j\sum_i^3T_iz_i^A\sigma_i\otimes\sigma_i],
\end{split}
\end{equation}
where the probabilities are
$$p_j^{bc}=\mathrm{Tr}((\Pi_j^A\otimes I_2\otimes I_2)\rho^{abc}(\Pi_j^A\otimes I_2\otimes I_2))=\frac{1}{2}[1+(-1)^ja_3z_3^A],$$
and $\sum_i^3(z^A_i)^2=1$. Therefore the reduced state of $\rho_j^{bc}$ is
$$\rho_j^b=\mathrm{Tr}_C\rho_j^{bc}=\frac{1}{2(1+(-1)^ja_3z_3^A)}[(1+(-1)^ja_3z_3^A)I_2+b_3\sigma_3+(-1)^j\sum_i^3r_iz_i^A\sigma_i]$$
with the probability $p_j^b=p_j^{bc}=\frac{1}{2}[1+(-1)^ja_3z_3^A]$. The eigenvalues of $\rho_j^b$ are
$$\lambda^{\pm}_j=\frac{1}{2(1+(-1)^ja_3z_3^A)}[1+(-1)^ja_3z_3^A\pm\sqrt{(b_3+(-1)^jr_3z_3^A)^2+\sum_i^2(r_iz_i^A)^2}].$$
We define the following entropy function
\begin{equation}
\begin{split}
H_\varepsilon(x)=\frac{1}{2}[(1+\varepsilon+x)\log_2(1+\varepsilon+x)+(1+\varepsilon-x)\log_2(1+\varepsilon-x)].
\end{split}
\end{equation}
Then measured conditional entropy of $B$ subsystem can be obtained as \cite{OZ,L,JY,V,GPW,SW}
\begin{equation}
\begin{split}
S_{B|\Pi^A}(\rho)&=-\sum_{j}p_j^b({\lambda^{+}_j \log_2\lambda^{+}_j}+\lambda^{-}_j \log_2\lambda^{-}_j)\\
&=-\frac{1}{2}[H_{a_3z_3^A}(A_+)+H_{-a_3z_3^A}(A_-)-2H(a_3z_3^A)-2],
\end{split}
\end{equation}
where $A_\pm=\sqrt{(b_3\pm r_3z_3^A)^2+\sum_i^2(r_iz_i^A)^2}$.

After measurement $\Pi_{k|j}^B$ on $BC$ system, the state $\rho_j^{bc}$ is changed to
\begin{equation}
\begin{split}
\rho_{jk}^{c}=&\frac{1}{p_{jk}^c}[(1+(-1)^ja_3z_3^A+(-1)^kb_3z_3^B+(-1)^{j+k}\sum_i^3r_iz_i^Az_i^B)I_2
\\+&\sum_i^3(c_i+(-1)^js_iz_i^A+(-1)^{k+j}T_iz_i^Az_i^B)\sigma_i], (j,k=0,1)
\end{split}
\end{equation}
with the probability  ($k=0,1$)
\begin{equation}
p_{0k}^{c}=\frac{1}{2(1+a_3z_3^A)}(1+\alpha_k),
 \ \  p_{1k}^{c}=\frac{1}{2(1-a_3z_3^A)}(1+\beta_k),
\end{equation}
where
\small$\alpha_k=a_3z_3^A+(-1)^k(b_3z_3^B+\sum_i^3r_iz_i^Az_i^B), \beta_k=-a_3z_3^A+(-1)^k(b_3z_3^B-\sum_i^3r_iz_i^Az_i^B)$.
The non-zero eigenvalues of $\rho^c_{jk}$ are given by
\begin{equation}\label{values}
\lambda_{0k}^{\pm}=\frac{1}{2(1+\alpha_k)}(1+\alpha_k\pm\gamma_k), \ \ \lambda_{1k}^{\pm}=\frac{1}{2(1+\beta_k)}(1+\beta_k\pm\delta_k), k=0,1,
\end{equation}
where
\begin{equation*}
\begin{split}
\gamma_k&=[\sum_i^3(c_i+s_iz_i^A+(-1)^kT_iz_i^Az^B_i)^2]^{\frac{1}{2}},\\
\delta_k&=[\sum_i^3(-c_i+s_iz_i^A+(-1)^kT_iz_i^Az^B_i)^2]^{\frac{1}{2}}.
\end{split}
\end{equation*}
According to the fact that the eigenvalues in Eq.(\ref{values}) are nonnegative, we have $\sqrt{\sum_i^3a_i^2}+\sqrt{\sum_i^3b_i^2}+\sqrt{\sum_i^3r_i^2}\leq 1$.

The entropy of $\rho^{abc}$ under the measurement $\Pi^{AB}$ is given by
\begin{equation}
\begin{split}
S_{C|\Pi^{AB}}(\rho)=&-\sum_{j,k} p^c_{jk}(\la_{jk}^{+}\log_2\la_{jk}^{+}+\la_{jk}^{-}\log_2\la_{jk}^{-})\\
=&-\frac{1}{2(1+a_3z_3^A)}[H_{\alpha_0}(\gamma_0)+H_{\alpha_1}(\gamma_1)-2H_{a_3z_3^A}(\frac{\alpha_0-\alpha_1}{2})]\\
&-\frac{1}{2(1-a_3z_3^A)}[H_{\beta_0}(\delta_0)+H_{\beta_1}(\delta_1)-2H_{-a_3z_3^A}(\frac{\beta_0-\beta_1}{2})]+2.
\end{split}
\end{equation}
In particularly, $a_3z_3^A=\frac{\alpha_0+\alpha_1}{2}$.

Let $G(z_1^A,z_2^A,z_3^A)=1-S_{B|\Pi^A}(\rho)$ and $F(z_1^A,z_2^A,z_3^A,z^B_1,z^B_2,z^B_3)=2 -S_{C|\Pi^{AB}}(\rho)$, then we have the following result.
\begin{theorem}\label{Th1}
For the non-X-states $\rho$ in Eq(\ref{eq:3.1}) with 14 parameters, the quantum discord is given by
\begin{equation}
\begin{split}
\mathcal{Q}(\rho)
=&-S_{ABC}(\rho)+S_{A}(\rho)+\min\{S_{B|\Pi^A}(\rho)+S_{C|\Pi^{AB}}(\rho)\}\\
=&3+\sum_{i=1}^8{\lambda_i\log_2\lambda_i}-\sum_{k=1}^2\lambda_k^{a}\log_2\lambda_k^{a}-\max_{z^X_i\in[0,1],\sum_i (z_i^X)^2=1}\{G+F\},
\end{split}
\end{equation}
where $\lambda_i (i=1,\cdots,8)$ are the eigenvalues of $\rho^{abc}$, $\lambda_k^{a}=\frac{1}{2}[1+(-1)^k a_3], (k=0,1) $ are eigenvalues of $\rho^{abc}$ on subsystem $A$ and $X$ represents subsystem $A, B$.
\end{theorem}

\begin{theorem}\label{Th2}
Let $r=\max{\{|r_1|,|r_2|\}}$, then $\max_{z^X_i\in[0,1],\sum_i (z_i^X)^2=1}\{G+F\}$ can be explicitly computed as follows.

Case1: when $a_3b_3r_3\leq0,r_3^2-r^2\geq a_3b_3r_3$, and $(b_3+r_3)(c_3+s_3)\leq 0$, we have
\begin{equation}
\max_{z^X_i\in[0,1],\sum_i (z_i^X)^2=1}\{G+F\}=G(0,0,1)+F(0,0,1,0,0,1),
\end{equation}
where
\begin{equation}
\begin{split}
G(0,0,1)=\frac{1}{2}[H_{a_3}(|b_3+r_3|)+H_{-a_3}(|b_3-r_3|)-2H(a_3)]
\end{split}
\end{equation}
and
\begin{equation}
\begin{split}
F(0,0,1,0,0,1)=&
\frac{1}{2(1+a_3)}[H_{\alpha_0}(\gamma_0)+H_{\alpha_1}(\gamma_1)-2H_{a_3}(b_3+r_3)]\\
+&\frac{1}{2(1-a_3)}[H_{\beta_0}(\delta_0)+H_{\beta_1}(\delta_1)-2H_{-a_3}(b_3-r_3)].
\end{split}
\end{equation}
In this case, the parameters are degenerated into $(k=0,1)$
\small$$\alpha_k=a_3+(-1)^k(b_3+r_3), \gamma_k=[\sum_i^3c_i^2+s_3^2+T_3^2+2(c_3s_3+(-1)^k(c_3T_3+s_3T_3))]^{\frac{1}{2}},$$
\small$$\beta_k=-a_3+(-1)^k(b_3-r_3),\delta_k=[\sum_i^3c_i^2+s_3^2+T_3^2+2(-c_3s_3+(-1)^k(s_3T_3-c_3T_3))]^{\frac{1}{2}}.$$

Case 2: (1) When $b_3=0, c_1s_1\leq0, s_1\leq|c_1|$ and $\max\{|r_1|,|r_2|,|r_3|\}=|r_1|$, we have
\begin{equation}
\max_{z^X_i\in[0,1],\sum_i (z_i^X)^2=1}\{G+F\}=G(1,0,0)+F(1,0,0,1,0,0),
\end{equation}
where
\begin{equation}
G(1,0,0)=\frac{1}{2}[H_{a_3}(r_1)+H_{-a_3}(r_1)-2H(a_3)]
\end{equation}
and
\begin{equation}
\begin{split}
F(1,0,0,1,0,0)=\frac{1}{2}[H_{r_1}(\gamma_0)+H_{-r_1}(\gamma_1)+H_{r_1}(\delta_0)+H_{-r_1}(\delta_1)-4H(r_1)].
\end{split}
\end{equation}
In this case, the parameters are degenerated into $(k=0,1)$
\small\begin{equation*}
\begin{split}
\gamma_k&=[\sum_i^3c_i^2+s_1^2+T_1^2+2(c_1s_1+(-1)^k(c_1T_1+s_1T_1))]^{\frac{1}{2}},\\
\delta_k&=[\sum_i^3c_i^2+s_1^2+T_1^2+2(-c_1s_1+(-1)^k(c_1T_1-s_1T_1))]^\frac{1}{2}.
\end{split}
\end{equation*}

(2) When $b_3=0, c_1s_1\leq0, s_1\leq|c_1|$ and $\max\{|r_1|,|r_2|,|r_3|\}=|r_2|$, we have
\begin{equation}
\max_{z^X_i\in[0,1],\sum_i (z_i^X)^2=1}\{G+F\}=G(0,1,0)+F(0,1,0,0,1,0),
\end{equation}
where
\begin{equation}
G(0,1,0)=\frac{1}{2}[H_{a_3}(r_2)+H_{-a_3}(r_2)-2H(a_3)]
\end{equation}
and
\begin{equation}
\begin{split}
F(0,1,0,0,1,0)=&\frac{1}{2}[H_{r_2}(\gamma_0)+H_{-r_2}(\gamma_1)+H_{r_2}(\delta_0)+H_{-r_2}(\delta_1)-4H(r_2)].
\end{split}
\end{equation}
In this case, the parameters are degenerated into $(k=0,1)$
\small\begin{equation*}
\begin{split}
\gamma_k&=[\sum_i^3c_i^2+s_2^2+T_2^2+2(c_2s_2+(-1)^kc_2T_2+s_2T_2)]^{\frac{1}{2}};\\
\delta_k&=[\sum_i^3c_i^2+s_2^2+T_2^2+2(-c_2s_2+(-1)^kc_2T_2-s_2T_2)]^{\frac{1}{2}}.
\end{split}\end{equation*}
\end{theorem}

\begin{proof}
By definition, we have
\begin{equation}
\begin{split}
G+F&=H_{a_3z_3^A}(B_+)+H_{-a_3z_3^A}(B_-)-2H(a_3z_3^A)\\
&+\frac{1}{2(1+a_3z_3^A)}[H_{\alpha_0}(\gamma_0)+H_{\alpha_1}(\gamma_1)-2H_{a_3z_3^A}(\frac{\alpha_0-\alpha_1}{2})]\\
&+\frac{1}{2(1-a_3z_3^A)}[H_{\beta_0}(\delta_0)+H_{\beta_1}(\delta_1)-2H_{-a_3z_3^A}(\frac{\beta_0-\beta_1}{2})].
\end{split}
\end{equation} 

Note that $F$ is a function of six variables and the first three are exactly the variables of $G$.
Our strategy of locating the extremal points of $G+F$ is first finding the critical points 
$z_1^A, z_2^A, z_3^A$ of $G$ and
verify that at those points the critical points of $G$ are attainable, then we can find the maximal points of $F+G$. %i.e. at this time $F$ has only three variables to consider.
%Under the condition that the function $G$ takes the maximum value, the maximum value of the function $F$ is obtained. 
%At the critical points $(z_1^A, z_2^A, z_3^A)$ of $F$, the function $F$ actually only has variables $z_1^B, z_2^B, z_3^B$.

For case 1: $a_3b_3r_3\leq 0$ and $r_3^2-r^2\geq a_3b_3r_3$, by \cite{JY} we know that $\max{G(z_1^A,z_2^A,z_3^A)}=G(0,0,1)$, then the parameters in function $F$ are degenerated into ($k=0,1$)
\begin{equation*}
\begin{split}\alpha_k&=a_3+(-1)^k(b_3z_3^B+r_3z_3^B),\\
\beta_k&=-a_3+(-1)^k(b_3z_3^B-r_3z_3^B),\\
\gamma_k&=\{\sum_i^3c_i^2+s_3^2+T_3^2(z^B_3)^2+2[c_3s_3+(-1)^k(s_3T_3z^B_3+c_3T_3z^B_3)]\}^{\frac{1}{2}},\\
\delta_k&=\{\sum_i^3c_i^2+s_3^2+T_3^2(z^B_3)^2+2[-c_3s_3+(-1)^k(s_3T_3z^B_3-c_3T_3z^B_3)]\}^{\frac{1}{2}}.
\end{split}
\end{equation*}
Therefore, we have 
\begin{equation}
\begin{split}
&F(z^A_1,z^A_2,z^A_3,z^B_1,z^B_2,z^B_3)=F(0,0,1,z^B_3)\\
&=\frac{1}{2(1+a_3)}[H_{\alpha_0}(\gamma_0)+H_{\alpha_1}(\gamma_1)]+\frac{1}{2(1-a_3)}[H_{\beta_0}(\gamma_0)+H_{\beta_1}(\gamma_1)].
\end{split}
\end{equation}
When $(b_3+r_3)(c_3+s_3)\leq 0$, it can be observed that $F$ is an even function for $z_3^B\in [-1,1]$, so we just need to consider $z_3^B\in [0,1]$. The derivative of $F$ on $z_3^B$ is given by
{\scriptsize\begin{equation}
\begin{split}
&\frac{\partial{F}}{\partial{z_3^B}}=\frac{1}{4(1+a_3)}\{(b_3+r_3)\log_2\frac{(1+\alpha_1)^2[(1+\alpha_0)^2-\gamma_0^2]}{(1+\alpha_0)^2[(1+\alpha_1)^2-\gamma_1^2]}\\
&+\frac{-c_3T_3-s_3T_3+T_3^2z_3^B}{\gamma_1}\log_2\frac{1+\alpha_1+\gamma_1}{1+\alpha_1-\gamma_1}
+\frac{c_3T_3+s_3T_3+T_3^2z_3^B}{\gamma_0}\log_2\frac{1+\alpha_0+\gamma_0}{1+\alpha_0-\gamma_0}\}\\
&+\frac{1}{4(1-a_3)}\{(b_3-r_3)\log_2\frac{(1+\beta_1)^2[(1+\beta_0)^2-\delta_0^2]}{(1+\beta_0)^2[(1+\beta_1)^2-\delta_1^2]}\\
&+\frac{c_3T_3-s_3T_3+T_3^2z_3^B}{\delta_1}\log_2\frac{1+\beta_1+\delta_1}{1+\beta_1-\delta_1}\
+\frac{-c_3T_3+s_3T_3+T_3^2z_3^B}{\delta_0}\log_2\frac{1+\beta_0+\delta_0}{1+\beta_0-\delta_0}\};
\end{split}
\end{equation}}

If $b_3+r_3\leq 0$ and $c_3+s_3\geq 0$, we have $\gamma_0\geq \gamma_1$, $\alpha_1\geq \alpha_0$, $\delta_0\geq \delta_1$ and $\beta_1\geq \beta_0$, then
{\scriptsize\begin{equation}
(b_3+r_3)\log_2\frac{(1+\alpha_1)^2[(1+\alpha_0)^2-\gamma_0^2]}{(1+\alpha_0)^2[(1+\alpha_1)^2-\gamma_1^2]}\geq 0;(b_3-r_3)\log_2\frac{(1+\beta_1)^2[(1+\beta_0)^2-\delta_0^2]}{(1+\beta_0)^2[(1+\beta_1)^2-\delta_1^2]}\geq 0;
\end{equation}
\begin{equation}
\begin{split}
&\frac{-c_3T_3-s_3T_3+T_3^2z_3^B}{\gamma_1}\log_2\frac{1+\alpha_1+\gamma_1}{1+\alpha_1-\gamma_1}
+\frac{c_3T_3+s_3T_3+T_3^2z_3^B}{\gamma_0}\log_2\frac{1+\alpha_0+\gamma_0}{1+\alpha_0-\gamma_0}\\
\geq& \frac{-c_3T_3-s_3T_3+T_3^2z_3^B}{\gamma_0}\log_2\frac{1+\alpha_1+\gamma_1}{1+\alpha_1-\gamma_1}
+\frac{c_3T_3+s_3T_3+T_3^2z_3^B}{\gamma_0}\log_2\frac{1+\alpha_1+\gamma_1}{1+\alpha_1-\gamma_1}\\
=&\frac{2T_3^2z^B_3}{\gamma_0}\log_2\frac{1+\alpha_1+\gamma_1}{1+\alpha_1-\gamma_1}\geq 0;
\end{split}
\end{equation}
\begin{equation}
\begin{split}
&\frac{c_3T_3-s_3T_3+T_3^2z_3^B}{\delta_1}\log_2\frac{1+\beta_1+\delta_1}{1+\beta_1-\delta_1}
+\frac{-c_3T_3+s_3T_3+T_3^2z_3^B}{\delta_0}\log_2\frac{1+\beta_0+\delta_0}{1+\beta_0-\delta_0}\\
\geq &\frac{c_3T_3-s_3T_3+T_3^2z_3^B}{\delta_0}\log_2\frac{1+\beta_1+\delta_1}{1+\beta_1-\delta_1}
+\frac{-c_3T_3+s_3T_3+T_3^2z_3^B}{\delta_0}\log_2\frac{1+\beta_1+\delta_1}{1+\beta_1-\delta_1}\\
=&\frac{2T_3^2z^B_3}{\delta_0}\log_2\frac{1+\beta_1+\delta_1}{1+\beta_1-\delta_1}\geq 0.
\end{split}
\end{equation}}
Hence in this case we get $\frac{\partial{F}}{\partial{z^B_3}}\geq0$ when $z_3^B\in[0,1]$.

If $b_3+r_3\geq 0$ and $c_3+s_3\leq 0$, we also can show that $\frac{\partial F}{\partial z_3^B}\geq 0$ similarly. So $F$ is a strictly monotonically increasing function with $z^B_3\in[0,1]$.
Similarly we can check that $F$ is a strictly monotonically increasing function with respect to $z^B_1\in[0,1]$ or $z^B_2\in[0,1]$ in case 2.
\end{proof}

\begin{theorem}\label{Th4}For the Werner-GHZ state $\rho_w = c|\psi\rangle\langle\psi|+(1-c)\frac{I}{8}$, where $|\psi\rangle=\frac{|000\rangle+|111\rangle}{2}$, the quantum discord is
\begin{equation}
\begin{split}
\mathcal{Q}=\frac{1}{8}(1-c)log_2(1-c)+\frac{1+7c}{8}log_2(1+7c)-\frac{1}{4}(1+3c)log_2(1+3c).
\end{split}
\end{equation}
\end{theorem}

\begin{proof}
Obviously, $\max{\{G(z_1^A,z_2^A,z_3^A)\}}=H(c).$ Let $\theta=cz_3^B$, then
\begin{equation}
\begin{split}
&F(z_1^A,z_2^A,z_3^A,z_1^B,z_2^B,z_3^B)=F(\theta)\\
=&\frac{1}{2}[H_{\theta}(|c+\theta|)+H_{\theta}(|c-\theta|)+H_{-\theta}(|c+\theta|)+H_{-\theta}(|c-\theta|)]-2H(\theta).
\end{split}
\end{equation}
It is easy to see that $F(\theta)$ is monotonically increasing with respect to $\theta\in[0,1]$. So $\max{\{F(\theta)\}}=F(\max{\{\theta\}})=F(c)$.
Fig.1 shows the behavior of the function $\mathcal{Q}$.
\begin{figure}[!htb]%[!t]
\centering
\includegraphics[width=3.5in]{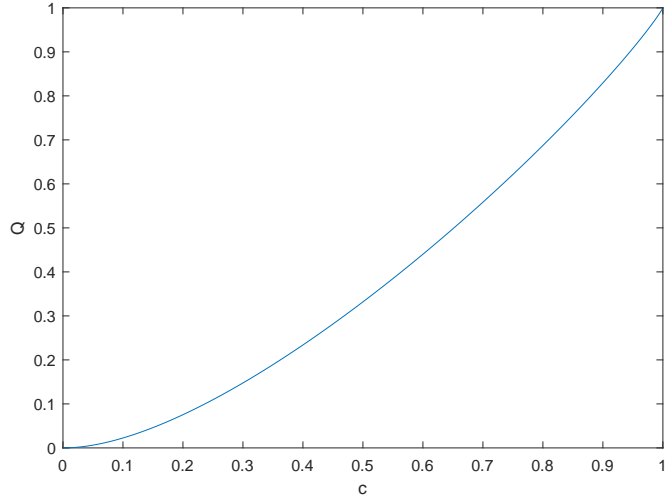}
{\small{\caption{The behavior of the quantum discord $\mathcal{Q}$ for the Werner-GHZ state in Theorem \ref{Th4}.}}}
\label{fig:1}
\end{figure}
\end{proof}

Next, we consider the following general tripartite state
\begin{equation}\label{eq:3.27}
\begin{split}
\rho&=\frac{1}{8}(I_8+\sum_i^3a_i\sigma_i\otimes I_4+I_2\otimes \sum_i^3b_i\sigma_i\otimes I_2+I_4\otimes \sum_i^3c_i\sigma_i\\&+\sum_i^3r_i \sigma_i\otimes \sigma_i\otimes I_2+\sum_i^3s_i\sigma_i\otimes I_2\otimes \sigma_i\\
&+\sum_i^3v_iI_2\otimes \sigma_i\otimes \sigma_i+\sum_i^3T_i\sigma_i\otimes \sigma_i\otimes \sigma_i).
\end{split}
\end{equation}
Let $a=\sqrt{\sum_i^3a_i^2}$ and $b=\sqrt{\sum_i^3b_i^2}$,
then we can get the quantum discord for some special cases.
\begin{theorem} \label{Th3}
For the general tripartite state $\rho$ in Eq.(\ref{eq:3.27}), we have the following results:

Case 1: when $a_i=v_i=T_i=0,r_1=r_2=r_3=r$, we have that
\begin{equation}
\begin{split}
\mathcal{Q}(\rho)=&\sum_i^8{\lambda_i\log_2\lambda_i}+4+H_{b}(r)+H_{-b}(r)-H(|b+r|)\\
-&\frac{1}{2}[H_{b+\mathrm{B}}(r)+H_{b-\mathrm{B}}(r)+H_{-b+\mathrm{B}}(r)+H_{-b-\mathrm{B}}(r)],
\end{split}
\end{equation}
where $\mathrm{B}=[\sum_i^3(s_i\frac{b_i}{b}+c_i)^2]^{\frac{1}{2}}$.

Case 2: when $b_i=v_i=T_i=0,r_1=r_2=r_3=r$, we have that
\begin{equation}
\begin{split}
\mathcal{Q}(\rho)&=\sum_i^8{\lambda_i\log_2\lambda_i}+3-H(a^2)
-\frac{1}{2}[H_{a}(r)+H_{-a}(r)-2H(a)]\\
&-\frac{1}{2(1+a)}[H_{a+\mathrm{A}}(r)+H_{a-\mathrm{A}}(r)-2H_{a}(r)]\\
&-\frac{1}{2(1-a)}[H_{-a+\mathrm{A}}(r)+H_{-a-\mathrm{A}}(r)-2H_{-a}(r)],
\end{split}
\end{equation}
where $\mathrm{A}=[\sum_i^3(s_i\frac{a_i}{a}+c_i)^2]^\frac{1}{2}$.

Case 3: when $r_i=T_i=v_i=0$, we have that
\begin{equation}
\begin{split}
\mathcal{Q}(\rho)&=\sum_i^8{\lambda_i\log_2\lambda_i}+3-H(a^2)
-\frac{1}{2}[H_{b}(a)+H_{-b}(a)-2H(a)]\\
&-\frac{1}{2(1+a)}[H_{a+\mathrm{A}}(b)+H_{a-\mathrm{A}}(b)-2H_{a}(b)]\\
&-\frac{1}{2(1-a)}[H_{b+\mathrm{A}}(b)+H_{-a-\mathrm{A}}(b)-2H_{-a}(b)],
\end{split}
\end{equation}
where $\mathrm{A}=[\sum_i^3(s_i\frac{a_i}{a}+c_i)^2]^\frac{1}{2}$.

Case 4: when $a_i=c_i=s_i=T_i=0,r_1=r_2=r_3=r, v_1=v_2=v_3=v$, we have that
\begin{equation}
\begin{split}
\mathcal{Q}(\rho)=&\sum_i^8{\lambda_i\log_2\lambda_i}+H_{b}(r)+H_{-b}(r)+4-H(|b+r|)\\
-&\frac{1}{2}[H_{b+v}(r)+H_{b-v}(r)+H_{-b+v}(r)H_{-b-v}(r)].
\end{split}
\end{equation}

Case 5: when $r_i=T_i=s_i=c_i=0, v_1=v_2=v_3=v$, we have that
\begin{equation}
\begin{split}
\mathcal{Q}(\rho)&=\sum_i^8{\lambda_i\log_2\lambda_i}+3-H(a^2)-\frac{1}{2}[H_{b}(a)+H_{-b}(a)-2H(a)]\\
&-\frac{1}{2(1+a)}[H_{a+v}(b)+H_{a-v}(b)-2H_{a}(b)]\\
&-\frac{1}{2(1-a)}[H_{-a+v}(b)+H_{-a-v}(b)-2H_{-a}(b)].
\end{split}
\end{equation}

Case 6: when $b_i=s_i=c_i=T_i=0, r_1=r_2=r_3=r, v_1=v_2=v_3=v$, we have that
\begin{equation}
\begin{split}
\mathcal{Q}(\rho)&=\sum_i^8{\lambda_i\log_2\lambda_i}+3-H(a^2)-\frac{1}{2}[H_{a}(r)+H_{-a}(r)-2H(a)]\\
&-\frac{1}{2(1+a)}[H_{a+v}(r)+H_{a-v}(r)-2H_{a}(r)]\\
&-\frac{1}{2(1-a)}[H_{-a+v}(r)+H_{-a-v}(r)-2H_{-a}(r)].
\end{split}
\end{equation}
\end{theorem}

\begin{proof}
All cases can be shown similarly. Let's consider case 1: $a_i=v_i=T_i=0,r_1=r_2=r_3=r$, %set $b=\sqrt{\sum_i^3b_i^2}$, then
$\max{\{G(z_1^A,z_2^A,z_3^A)\}}=G(\frac{b_1}{b},\frac{b_2}{b},\frac{b_3}{b})$. Let $\theta=\sum_i^3rz_i^B\frac{b_i}{b}$, then
\begin{equation}
\begin{split}
&F(z_1^A,z_2^A,z_3^A,z_1^B,z_2^B,z_3^B)=F(\frac{b_1}{b},\frac{b_2}{b},\frac{b_3}{b},\theta)\\
=&\frac{1}{2}[H_{b+\mathrm{B}}(\theta)+H_{b-\mathrm{B}}(\theta)+H_{-b+\mathrm{B}}(\theta)+H_{-b-\mathrm{B}}(\theta)]-H_{b}(\theta)-H_{-b}(\theta)-2,
\end{split}
\end{equation}
where $\mathrm{B}=[\sum_i^3(s_i\frac{b_i}{b}+c_i)^2]^{\frac{1}{2}}$.

The derivative of $F$ over $\theta$ is equal to
\begin{equation}
\begin{split}
\frac{\partial F}{\partial\theta}=&\frac{1}{4}[\log_2\frac{(1+b+\mathrm{B}+\theta)(1+b-\mathrm{B}+\theta)(1+b-\theta)^2}{(1+b+\mathrm{B}-\theta)(1+b-\mathrm{B}-\theta)(1+b+\theta)^2}\\
+&\log_2\frac{(1-b-\mathrm{B}+\theta)(1-b+\mathrm{B}+\theta)(1-b-\theta)^2}{(1-b-\mathrm{B}-\theta)(1-b+\mathrm{B}-\theta)(1-b+\theta)^2}].
\end{split}
\end{equation}
Obviously, $\frac{\partial F}{\partial \theta}\geq0$ when $\theta\in[0,1]$. Then $F(\theta)$ is a strictly increasing function and $\max{F(\theta)}=F(\max\{\theta\})$.

Let $Y=\theta+\mu[1-(z_1^B)^2-(z_2^B)^2-(z_3^B)^2]$,
$\frac{\partial Y}{\partial z_1^B}=r\frac{b_1}{b}-2\mu z_1^B$, $\frac{\partial Y}{\partial z_2^B}=r\frac{b_2}{b}-2\mu z_2^B$, $\frac{\partial Y}{\partial z_3^B}=r\frac{b_3}{b}-2\mu z_3^B$, $\frac{\partial Y}{\partial \mu}=1-(z_1^B)^2-(z_2^B)^2-(z_3^B)^2$.
Imposing $\frac{\partial Y}{\partial z_1^B}=0, \frac{\partial Y}{\partial z_2^B}=0, \frac{\partial Y}{\partial z_3^B}=0, \frac{\partial Y}{\partial \mu}=0$, we have $z_i^B=\frac{b_i}{b}$. So $\max\{\theta\}= r$ and $\max{F(\theta)}=F(r)$, then case 1 is shown.
\end{proof}

Example 1. For a state in Eq.(\ref{eq:3.1}), when $a_1=0, a_2= 0, a_3 = 0.03, b_1 = 0, b_2 = 0, b_3 = 0.25, c_1 = 0.12, c_2= 0.12, c_3 = 0.01, r_1 = 0.1, r_2 = 0.1, r_3 = -0.3, s_1= 0.13, s_2 = 0.13, s_3 =-0.26, v_1 = 0, v_2 = 0, v_3 = 0, T_1 = -0.02, T_2 = -0.02, T_3 = -0.36$.
%and the eigenvalues are $\lambda_1=0.28803, \lambda_2=0.191856, \lambda_3=0.161101, \lambda_4=0.132845, \lambda_5=0.0933099, \lambda_6=0.0702772, \lambda_7=0.0432765, \lambda_8=0.0193042$.
According to the case 1 of Theorem \ref{Th2}, we have $\mathcal{Q}=0.8889$. Fig. 2 shows the behavior of the quantum discord $\mathcal{Q}$.
\begin{figure}[!t]
\centering
\includegraphics[width=3.5in]{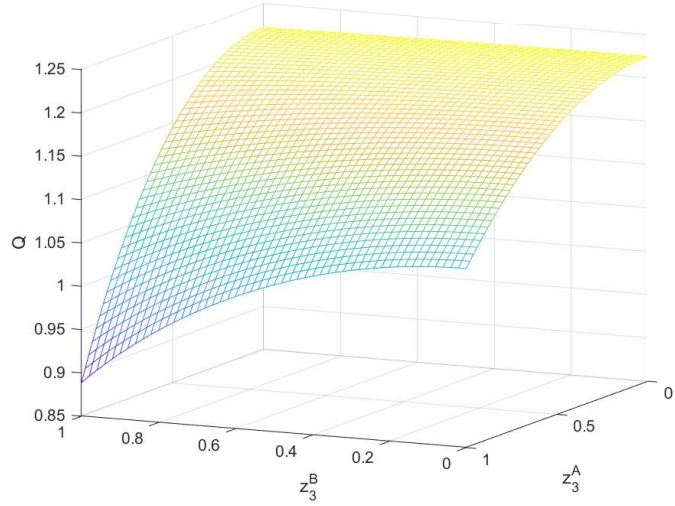}
\small{\caption{The behavior of the quantum discord $\mathcal{Q}$ with respect to the parameters in Example 1. In this case, $z_1^A=z_2^A=z_1^B=z_2^B=0$, the quantum discord $\mathcal{Q}$ is only related to variables $z_3^A, z_3^B$. Then $\mathcal{Q}=0.8889$.}}
\label{fig:2}
\end{figure}

Example 2. For a state of the case 1 in Theorem \ref{Th3}, when $a_1=a_2=a_3 = 0, b_1 = 0.2, b_2 = 0.05, b_3 = 0.1, c_1 = 0.04, c_2= 0.06, c_3 = 0.11, r_1 = r_2 = r_3 = 0.17, s_1= 0.08, s_2 = 0.15, s_3 =0.25,
v_1 = v_2 = v_3 = T_1 = T_2 = T_3 = 0$.
%the eigenvalues are $\lambda_1=0.2177, \lambda_2=0.1827, \lambda_3=0.1644, \lambda_4=0.1274, \lambda_5=0.1235, \lambda_6=0.1132, \lambda_7=0.0614, \lambda_8=0.0097$.
Then the quantum discord is $\mathcal{Q}=0.9970$. Fig. 3 and Fig. 4 show the behavior of the function $G$ and $F$ respectively.
\begin{figure}[!t]
\centering
\includegraphics[width=4.5in]{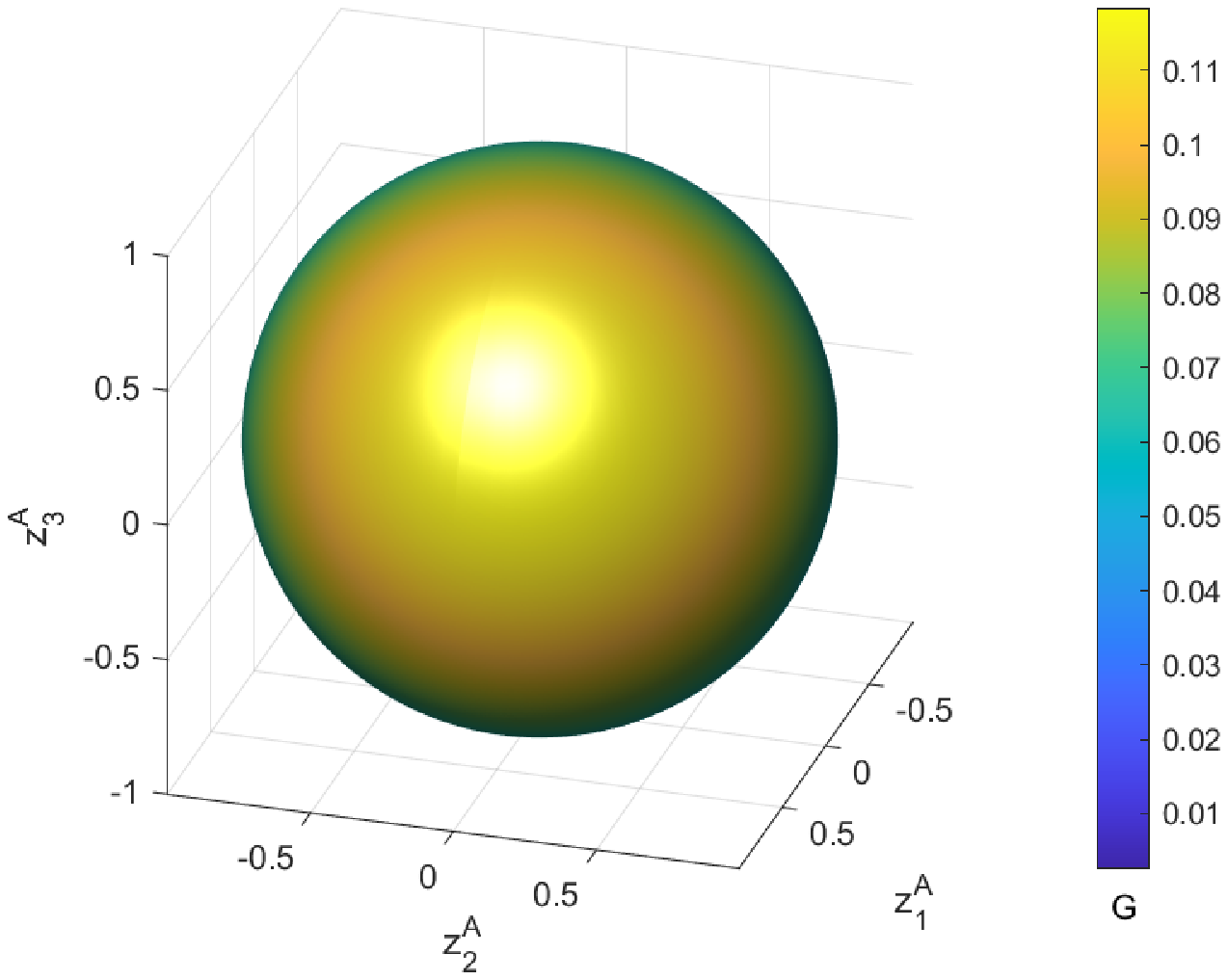}
\small{\caption{ The behavior of $G(z_1^A, z_2^A, z_3^A)$ with the variables $z_1^A, z_2^A, z_3^A$ in Example 2, where $z^A=(z_1^A, z_2^A, z_3^A)$ is on a unit sphere. This is a four-dimensional figure. Among them, the intensity of light is used to indicate the magnitude of the $G$ value. The brighter the point, the greater the value of $G$ and $\max G=0.1182$.}}
\label{fig:3}
\end{figure}

\begin{figure}[!t]
\centering
\includegraphics[width=4.5in]{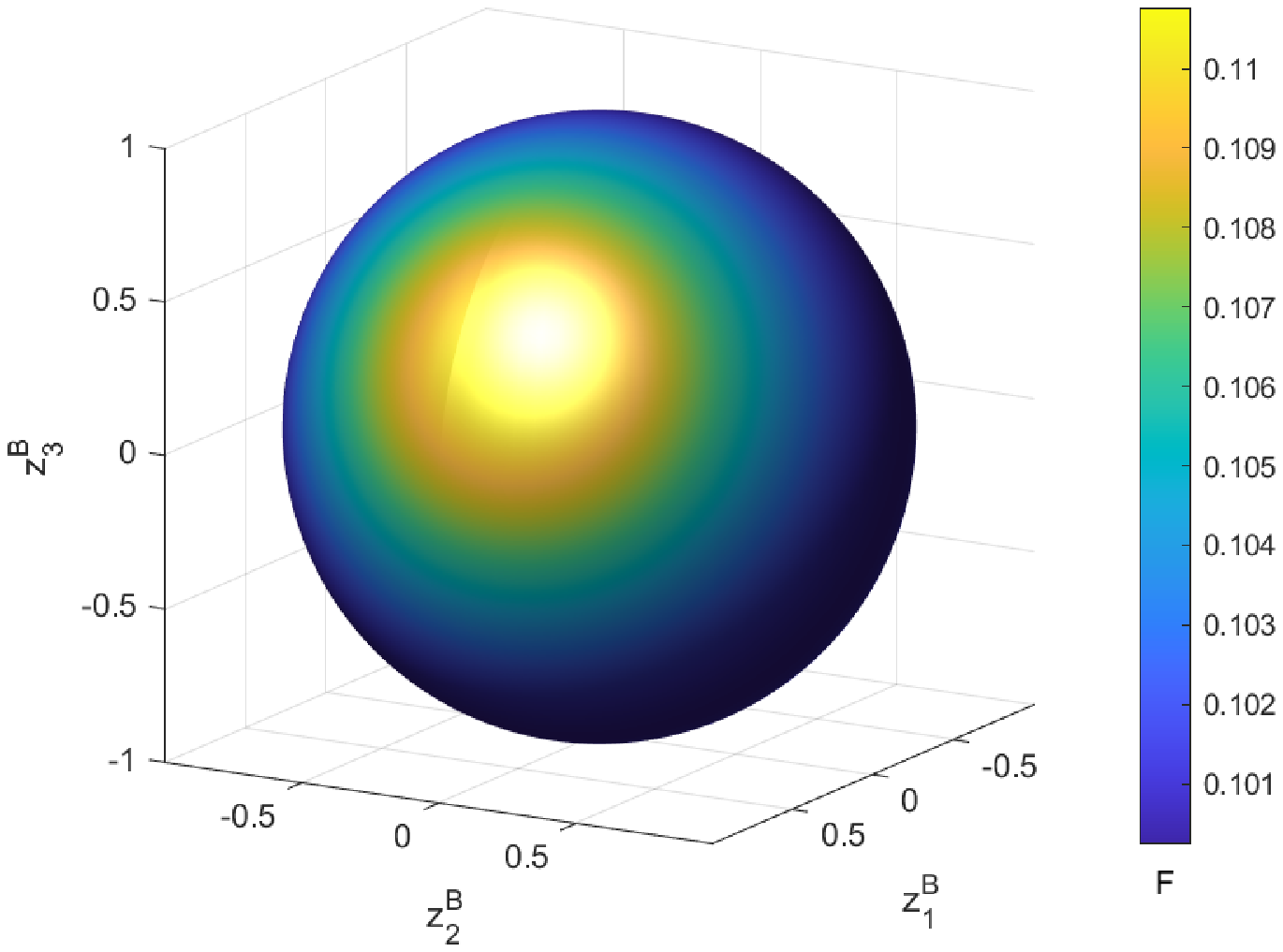}
\small{\caption{The behavior of $F(0.8729, 0.2182, 0.4364, z_1^B, z_2^B, z_3^B)$ with the variables $z_1^B, z_2^B, z_3^B$ in Example 2, where $z^B=(z_1^B, z_2^B, z_3^B)$ is on a unit sphere. The brighter the point, the greater the value of $F$ and $\max F=0.1107$.}}
\label{fig:4}
\end{figure}

\section{Conclusions}
Quantum discord is one of the important correlations in studying quantum systems. It is well-known that the quantum discord is hard to compute explicitly, and only sporadic formulas are known, for instance, the Bell state and the X-state etc. Recently important progresses are made to generalize the notion to multipartite quantum systems \cite{RLB}, and their explicit formulas are expectedly not easy to find. In this work, we have found explicit formulas of the quantum discord for tripartite non X-states with 14 parameters, including some famous states such as the Werner-GHZ state.
\vskip 1in

\centerline{\bf Acknowledgments}
\medskip
The research is supported in part by the NSFC
grants 11871325 and 12126351, and Natural Science Foundation of Hubei Province
grant no. 2020CFB538 as well as Simons Foundation
grant no. 523868.
\vskip 0.1in

\bibliographystyle{amsalpha}

\end{document}